\documentclass[10pt,reqno]{amsart}
\usepackage{amsthm}
\usepackage{amsmath}
\usepackage{latexsym}
\usepackage{amsfonts}
\usepackage{amssymb}
\usepackage{color}
\usepackage{bbm,dsfont}
\usepackage{graphicx}
\usepackage{hyperref}
\usepackage{enumerate}
\usepackage{comment}
\usepackage{float}
\usepackage{phfqit}

\newtheorem{theorem}{Theorem}
\newtheorem{lemma}{Lemma}

\theoremstyle{definition}

\def\subtitle#1{\gdef\@subtitle{#1}}
\def\@subtitle{}
\makeatother

\begin{document}\setlength{\arraycolsep}{2pt}

\title[]{\textbf{Locally unitary quantum state evolution is local}}

\author[Heikkilä]{Matias Heikkilä}
\email{Matias Heikkilä: mapehe@iki.fi}

\begin{abstract}
  We study the localization properties of bipartite channels, whose action on a
  subsystem yields a unitary channel. In particular we show that, under such
  channel, the subsystem must evolve independent of its environment. This point
  of view is another way to verify certain well-known conservation laws of
  quantum information in a generalized way. A no-go theorem for non classical
  conditional semantics in quantum computation is obtained as an intermediate
  result.
\end{abstract}

\maketitle

\section{Introduction}
Two important features of the standard formulation of the quantum theory are
the reversible dynamic evolution of isolated systems and quantum nonlocality.
Reversibility of physical processes, formally expressed as unitarity of state
transformations, is closely related to the compatibility of quantum mechanics
with other physical theories, the black hole information paradox being a
well-known example of this dissonance \cite{Hawking1976}. Quantum nonlocality,
on the other hand, has prompted several foundational contributions, including
the well-known EPR paradox \cite{Einstein1935}, and Bell's later argument
\cite{Bell1964} against certain hidden-variable theories. In this paper, we
seek to clarify the connection between these two topics, by showing, in
particular, that unitarity implies locality of a physical process. This also
turns out to be a simple point of view to verify some of the several
conservation laws of quantum information, several of which have been developed
in the literature \cite{Wootters1982} \cite{Caves1996} \cite{Pati2000}
\cite{Braunstein2007}.

The formal notion of a physical process is a quantum channel, a completely
positive, trace-preserving linear map on the state space. In
\cite{Beckman2001}, the authors propose a set of localization properties of
these maps from relativistic considerations, and the relationship between these
classes of quantum channels is further clarified in \cite{Eggeling2002}. In
this paper, we further apply the methods of \cite{Eggeling2002} to obtain the
present results.

In a bipartite system, entangled states \cite{Horodecki2009} allow for a
nonclassical destination for local information: As a simple example, consider a
unitary channel, that transforms a separable pure state to a maximally
entangled \cite{Horodecki2001} state. Entanglement is a major source a
nontriviality, and gives rise to many interesting nonclassical phenomena such
as steering \cite{Uola2020}, already discussed in Schrödinger's early works
\cite{Schrodinger1935} \cite{Schrodinger1936}, and the more recently discovered
state teleportation \cite{Bennett1993} and superdense coding
\cite{Bennett1984}. Even more recently, entanglement has been identified as a
quantum resource \cite{Chitambar2019}, and analogues of entanglement witnesses
have been studied \cite{Araujo2015} \cite{Hoban2015} in the context of
generalizations of quantum theory where the causal order between various
operations is not defined \cite{Oreshkov2012} \cite{Chiribella2013}
\cite{Procopio2015} \cite{Zych2019} \cite{Kristjnsson2020} \cite{Milz2022}.

One often considers less than general physical processes such as local, LOCC
\cite{Chitambar2014} or LOSR \cite{Buscemi2012} \cite{Schmid2020} channels.
Sufficiently strong localization properties rule out the possibility of a
separable state to become entangled. As the localization properties of quantum
channels influence their ability to create entanglement, there is some
conceptual justification why localization properties should also be related to
the local conservation of information, and, hence, some intuitive support for
the present results.

This paper is organized as follows: We present our results in Section
\ref{section::presentation}. Concluding remarks are given in Section
\ref{section::conclusion}. The technical material, including the proofs,
is appears in Appendix \ref{section::main}.

\section{On unitarity and locality}
\label{section::presentation}

Consider a finite-dimensional complex Hilbert space $\mathcal{H}$ with
$\braket{\psi}{\lambda \varphi + \mu \phi}= \lambda \braket{\psi}{\varphi} +
\mu \braket{\psi}{\phi}$ and $\braket{\psi}{\varphi} =
\overline{\braket{\varphi}{\psi}}$. A positive operator $\rho$ is a
\emph{state} if $\tr \left[ \rho \right] = 1$. Recall the canonical notation
\cite{Dirac1927}: For $\psi, \varphi \in \mathcal{H}$, $\ketbra{\psi}{\varphi}$
denotes the operator $\xi \mapsto \braket{\varphi}{\xi} \psi$, and one may
interpret $\ket{\psi} \in \mathcal{H}$, $\bra{\varphi} \in \mathcal{H}^*$. In
particular, $\ketbra{\psi}{\psi}$ is the projection onto the one-dimensional
subspace spanned by $\psi$. A completely positive, linear and trace-preserving
mapping on the operators of $\mathcal{H}$ is called a \emph{channel}. Quantum
channels describe the physically realizable ways for the state of a system to
evolve. A channel $\mathcal{N}$ is a unitary channel, if $\mathcal{N}
\left(\rho \right) = U \rho U^*$ for some unitary operator $U$. A channel
$\mathcal{N}$ has a dual $\mathcal{N}^*$ via $\tr \left[ \mathcal{N} \left(
\rho \right) E \right] = \tr \left[\rho \mathcal{N}^* \left( E \right)\right]$.
These points of view are referred to as the Schrödinger picture and the
Heisenberg picture, and $\mathcal{N}^*$ is occasionally called Heisenberg dual
of $\mathcal{N}$. A bipartite channel $\mathcal{N}$ on $\mathcal{H}_A \otimes
\mathcal{H}_B$ is local, if $\mathcal{N} = \mathcal{N}_A \otimes \mathcal{N}_B$
for some channels $\mathcal{N}_A$ and $\mathcal{N}_B$ on $\mathcal{H}_A$ and
$\mathcal{H}_B$ respectively.

There are several theorems that describe the manner in which quantum
information is preserved in physical processes. A well-known example of such
result is the \emph{no-cloning theorem} \cite{Wootters1982}: There is no
quantum channel $\mathcal{N}$ that satisfies
\begin{equation}
  \label{eq::cloning}
\mathcal{N} \left(\rho \otimes \xi \right) = \rho \otimes \rho
\end{equation}
for all states $\rho$ on $\mathcal{H}_A$ and $\xi$ some fixed state on
$\mathcal{H}_B$. A more general version of this is the \emph{no-broadcasting
theorem} \cite{Caves1996} i.e. the fact that there is no channel $\mathcal{N}$
such that
$$
\mathcal{N} \left(\rho \otimes \xi \right) = \tilde \rho,
$$
with $\rho = \tr_A \left[ \tilde \rho \right] = \tr_B§ \left[ \tilde \rho
\right]$. In this paper, we first prove the following, which implies both of
these theorems in a finite dimension.

\begin{theorem}
  \label{theorem::theorem}
  Let $\mathcal{H}_A$ and $\mathcal{H}_B$ be finite-dimensional. Consider a
  fixed pure state $\xi$ on $\mathcal{H}_B$ and a channel $\mathcal{N}$ on
  $\mathcal{H}_A \otimes \mathcal{H}_B$. Define a channel $\mathcal{N}_A$ on
  $\mathcal{H}_A$ via
$$
\mathcal{N}_A \left( \rho \right) = \tr_B \left[ \mathcal{N} \left( \rho
  \otimes \xi \right) \right].
$$
  If $\mathcal{N}_A$ is unitary, then
$$
  \mathcal{N} \left( \rho \otimes \xi \right) = \mathcal{N}_A \left(
  \rho \right) \otimes \mathcal{N}_B \left( \xi \right)
$$
  for some channel $\mathcal{N}_B$ on $\mathcal{H}_B$.
\end{theorem}

As unitary transformations can be described as information-preserving, an
informal way to put Theorem \ref{theorem::theorem} is that if a physical
process on a composite system $AB$ preserves quantum information in system $A$,
there can be no transfer of information between the systems $A$ and $B$. This
is consistent with the fact that a unitary channel and its conjugate are
dynamically independent (see Example 4.41 \cite{Heinosaari2011}). There is also
a closely related result by Busch \cite{Busch2003}, according to which a
bipartite unitary channel that doesn't create any entangled states is
necessarily of a similar product form as $\mathcal{N}$ in Theorem
\ref{theorem::theorem}, or of a product form with an added application of SWAP
gate.

Note that the proposed cloning channel in Equation \eqref{eq::cloning}
satisfies 
$$
\rho = \tr_B \left[ \mathcal{N} \left( \rho \otimes \xi \right) \right],
$$
which corresponds to Theorem \ref{theorem::theorem} with $\mathcal{N}_A$ the
identity map. Then, according to Theorem \ref{theorem::theorem}, the final
state of the environment cannot depend on the state of this system, which is a
property it would need to have in order to clone or broadcast the state.
Theorem \ref{theorem::theorem} is then, in this sense, a generalization of both
results.

Note that it is still possible to clone orthogonal states via an extension of
the action
$$
\ket{x} \otimes \ket{0} \mapsto \ket{x} \otimes \ket{x},
$$
with $\left\{ \ket{x} \right\}$ an orthonormal basis of $\mathcal{H}_A$. The
partial trace of this, however, is the unconditional output state of a Lüders
instrument of a sharp observable associated with the orthonormal basis $\left\{
  \ket{x} \right\}$, which is not unitary. This is then at odds with Theorem
\ref{theorem::theorem}.

A natural continuation of Theorem \ref{theorem::theorem} is to ask, in what manner
are unitary channels corresponding to different choices of $\xi$ compatible.
Regarding that we prove the following:

\begin{theorem}
  \label{theorem::no_conditional}
  Let $\mathcal{H}_A$ and $\mathcal{H}_B$ be finite-dimensional and let
  $\mathcal{N}$ be a channel on $\mathcal{H}_A \otimes \mathcal{H}_B$. Let
  $\xi_1$ and $\xi_2$ be non-orthogonal pure states on $\mathcal{H}_A$. Define
  channels $\mathcal{N}_{A, 1}$ and $\mathcal{N}_{A, 2}$ on $\mathcal{H}_A$ via
$$
  \mathcal{N}_{A, i} \left( \rho \right) = \tr_B \left[ \mathcal{N} \left( \rho
  \otimes \xi_i \right) \right],
$$
for $i = 1, 2$. If both $\mathcal{N}_{A, 1}$ and $\mathcal{N}_{A, 2}$ are
  unitary, then $\mathcal{N}_{A, 1} = \mathcal{N}_{A, 2}$.
\end{theorem}

We note that, from the point of view of quantum computing, this can be seen as
a no-go theorem prohibiting non-classical conditional statements: In classical
computing, a simple conditional statement involves a boolean controller and two
possible branches, one of which is executed depending on the value of the
controller. A quantum analogue of this is a controlled gate \cite{Barenco1995},
which can be constructed as follows: Consider a finite-dimensional complex
Hilbert space $\mathcal{H}_A$ of a physical system, and a similarly defined
controller $\mathcal{H}_B$. Define the operators
$$
A_x = U_x \otimes \ketbra{x}{x},
$$
with $U_x$ some unitary operators on $\mathcal{H}_A$ and $\left\{ \ket{x}
\right\}$ some orthonormal basis on $\mathcal{H}_B$. Since $\sum A_x^*A_x = I$,
these $A_x$ are the Kraus operators of some quantum channel $\mathcal{N}$,
which we may interpret as the conditional execution of one of the programs
$U_x$. In other words, the classical analogue of the channel
$$
\rho \mapsto \tr_B \left[ \mathcal{N} \left( \rho \otimes \ketbra{x}{x} \right)
\right] = U_x \rho U^*_x,
$$
is a conditional statement that chooses to execute one $U_x$ depending on
the value of $x$, i.e. the state of the controller.

It is, however, a basic fact that the set of the pure states on $\mathcal{H}_B$
is more complicated than the set of states of a  classical $n$-bit register due
to the superposition structure. It's then a natural question to consider
whether this would allow for non-classical conditional semantics in quantum
algorithms. In other words, letting $\xi$ be any pure state on $\mathcal{H}_B$,
is it possible to construct the channel $\mathcal{N}$ so that the expressions
$\tr_B \left[ \mathcal{N} \left( \rho \otimes \xi \right) \right]$ with
nonorthogonal $\xi$ induce different unitary channels on $\mathcal{H}_A$?
Theorem \ref{theorem::no_conditional} answers this in negative, as it shows
that a $2^n$-dimensional quantum controller can support at most $2^n$ different
conditional routines $U_x$, as for any superposition of two basis vectors,
either both basis vectors represent the same choice of $U_x$ or $\tr_B \left[
  \mathcal{N} \left( \rho \otimes \xi \right) \right]$ fails to be unitary. In
this sense, the power of a $n$-qubit controller must then coincide with that of
the classical $n$-bit register.

Note that Theorem \ref{theorem::no_conditional} is another generalization of
the no-cloning theorem, in the sense that if one could use non-orthogonal
states $\xi_1$ and $\xi_2$ to induce channels on $\mathcal{H}_A$ corresponding
to freely chosen unitaries $U_1$ and $U_2$, such process could be used to
perfectly distinguish those states. It is, however, stricter as it not only
rules out the free choice of $U_1$ and $U_2$ but \emph{any} $U_1$ and $U_2$
that would induce different unitary channels on $\mathcal{H}_A$.

Finally, a combination of Theorems \ref{theorem::theorem} and
\ref{theorem::no_conditional} is that there is no information exchange between
the systems or creation of correlation unless non-unitary state evolution
is involved:

\begin{theorem}
  \label{theorem::theorem3}
  Let $\mathcal{H}_A$ and $\mathcal{H}_B$ be finite-dimensional. Consider a channel
  $\mathcal{N}$ on $\mathcal{H}_A \otimes \mathcal{H}_B$.
  If each channel $\mathcal{N}_A$
$$
  \mathcal{N}_{A, \xi} \left( \rho \right) = \tr_B \left[ \mathcal{N} \left(
  \rho \otimes \xi \right) \right],
$$
  with $\xi$ a pure state on $\mathcal{H}_B$, is unitary, then
$$
  \mathcal{N} = \mathcal{N}_A \otimes \mathcal{N}_B 
$$
  for a fixed unitary channel $\mathcal{N}_A$ on $\mathcal{H}_A$ and a channel
  $\mathcal{N}_B$ on $\mathcal{H}_B$.
\end{theorem}

This is related to Theorem 7 of \cite{Beckman2001}, according to which a
unitary semicausal channel is local. In particular, if $\mathcal{N}$ is unitary
semicausal, then $\mathcal{N}_{A, \xi}$ is unitary and independent of $\xi$,
and then, by Theorem \ref{theorem::theorem}, of the desired product form.
Theorem \ref{theorem::theorem3} proves another direction: If the local
evolution of the system is unitary for all states of the environment, these
systems evolve independently, and, in particular, the evolution of the
composite state is semicausal.

A consequence of Theorem \ref{theorem::theorem3} is that if Bob, preparing the
state on $\mathcal{H}_B$, has any sort of influence on the state of Alice's
system $\mathcal{H}_A$, it is always possible for Bob to construct a state that
forces Alice's system to evolve in a non-unitary manner: Even if
$\mathcal{N}_A$ was unitary for an orthogonal set of states on $\mathcal{H}_B$,
it can never be unitary for all of states, unless the state evolution on
$\mathcal{H}_A$ is completely independent of the state on $\mathcal{H}_B$. One
way to interpret this is that, under any circumstance where $\mathcal{H}_B$ has
influence on $\mathcal{H}_A$, it's not possible to avoid the transmission of
noise from Bob's system to Alice's if the state of the controller $\mathcal{H}_B$
is improperly prepared.

\section{Concluding remarks}
\label{section::conclusion}

Theorems \ref{theorem::theorem} and \ref{theorem::theorem3} are another
addition to the collection of conservation laws related to quantum information.
An interesting feature of these theorems is that they are not a no-go theorems
per se, rather they characterize the physically possible time evolutions under
certain circumstances. Some no-go theorems are then produced by noting that
certain types of state evolution are not of the permitted form. An interesting
direction for future work would be to further clarify the relation of the
present work to various no-go theorems not considered here. Another interesting
question is whether these theorems have a valid generalization to the infinite
dimension. That discussion would likely require more refined analysis however,
as the definition of the channel $\mathcal{N}^*_B$ in the proof of Theorem
\ref{theorem::theorem} doesn't make sense in the infinite dimension without
some additional work.

Theorem \ref{theorem::no_conditional} can either be viewed as an intermediate
step between  Theorems \ref{theorem::theorem} and \ref{theorem::theorem3}, or
it can be considered in isolation as a statement regarding the limits of
quantum algorithms. Such results are of immediate interest as quantum
computation is a nascent technology that has attracted significant interest
due to the developments during the past years, such as \cite{Arute2019}.

\section*{Acknowledgements}

The author wishes to thank Oskari Kerppo for his valuable feedback on an
earlier version of this manuscript.

\bibliographystyle{plain}
\bibliography{mybib}

\appendix

\section{Proofs}
\label{section::main}

In the context of quantum channels, the Stinespring dilation of a channel
$\mathcal{N}$ can be understood in the Heisenberg picture as
\begin{equation}
\label{eq::stinespring}
\mathcal{N}^* \left( F \right) = V^* \left( F \otimes I_E \right) V,
\end{equation}
for linear maps $F$ on $\mathcal{H}_A$, which involves some bounded operator $V
: \mathcal{H}_A \to \mathcal{H}_A \otimes \mathcal{H}_E$, $V^*V = I_A$ and the
unital *-homomorphism $\pi_E: F \mapsto F \otimes I_E$. A triple $\left\langle
\mathcal{H}_A \otimes \mathcal{H}_E, \pi_E, V \right\rangle$ such that the
Equation \eqref{eq::stinespring} holds is a Stinespring dilation of
$\mathcal{N}^*$.

In the finite-dimensional setting, a Stinespring dilation is called
\emph{minimal} if the vectors $\left(F \otimes I_E \right) V \psi$ span
$\mathcal{H}_A \otimes \mathcal{H}_E$. If $\left\langle \mathcal{H}_A \otimes
\mathcal{H}_{E_1}, \pi_{E_1}, V_1 \right\rangle$ is a minimal Stinespring
dilation of $\mathcal{N}^*$ and $\left\langle \mathcal{H}_A \otimes
\mathcal{H}_{E_2}, \pi_{E_2}, V_2 \right\rangle$ is another Stinespring
dilation of the same channel, the relation
\begin{equation}
  \label{eq::tilde_u}
\tilde{U} \left( F \otimes I_{E_1}\right) V_1 \psi = \left( F \otimes
I_{E_2}\right) V_2 \psi,
\end{equation}
for all $\psi \in \mathcal{H}_A$ and linear maps $F$ on $\mathcal{H}_A$,
produces a well-defined isometry $\tilde{U}: \mathcal{H} \otimes E_1 \to
\mathcal{H} \otimes E_2$ that satisfies the relation $\tilde{U} \left(F \otimes
1_{E_1} \right) = \left(F \otimes 1_{E_2} \right) \tilde{U}$ for all $F$ (see
the proof of Theorem 7.2 \cite{Busch2016}). As pointed out \cite{Eggeling2002},
we have the following. A proof is given for reference.

\begin{lemma}
  \label{lemma::decompose}
  The isometry $\tilde{U}$ in Equation \eqref{eq::tilde_u} can be decomposed as
  $\tilde{U} = I \otimes U_E$, with $U_E: \mathcal{H}_{E_1} \to
  \mathcal{H}_{E_2}$ an isometry. And, in particular, $V_2 = \left( I_A \otimes
  U_E \right) V_1$.
\end{lemma}
\begin{proof}
  Fix some unit vector $\nu \in \mathcal{H}$ and define the maps $V:
  \mathcal{H}_{E_1} \to \mathcal{H} \otimes \mathcal{H}_{E_1}$, $V \varphi =
  \nu \otimes \varphi$ and $W: \mathcal{H}_{E_2} \to \mathcal{H} \otimes
  \mathcal{H}_{E_2}$, $W \phi = \nu \otimes \phi$, and claim $\tilde{U} = 1
  \otimes W^* U V$.

  Noting that
  $$
  \braket{\psi_{E_2}}{U \psi_{E_1}} = \sum_{i, j, k, l} \lambda_{i, j}
  \overline{\mu_{k, l}} \braket{\psi_i \otimes \phi_j }{U \psi_k \otimes
  \varphi_l},
  $$
  with $\lambda_{i, j}, \mu_{k, l} \in \mathbb{C}$, $\psi_{E_1} \in
  \mathcal{H}_{E_1}$, $\psi_{E_2} \in \mathcal{H}_{E_2}$ and $\left\{ \psi_i
  \right\}$, $\left\{ \varphi_i \right\}$ and $\left\{ \phi_i \right\}$ some
  orthonormal bases of $\mathcal{H}$, $\mathcal{H}_{E_1}$ and
  $\mathcal{H}_{E_2}$ respectively, in order to prove that,
  $$
  \braket{\psi_{E_2}}{\tilde{U} \psi_{E_1}} = \braket{\psi_{E_2}}{\left( I
  \otimes W^* U V \right) \psi_{E_1}},
  $$
  it suffices to show
  $$
  \braket{\psi_i \otimes \phi_j }{\tilde{U} \psi_k \otimes \varphi_l} =
  \braket{\psi_i \otimes \phi_j }{\left( I \otimes W^* U V \right) \psi_k \otimes
  \varphi_l}.
  $$

  Let $i \neq k$. Since $\tilde{U} \left(F \otimes I_{E_1} \right) = \left(F
  \otimes I_{E_2} \right) \tilde{U}$ for all linear maps $F$ on
  $\mathcal{H}_A$,
  \begin{align*}
    & \braket{\psi_i \otimes \phi_j }{\tilde{U} \psi_k \otimes \varphi_l} =
  \braket{\psi_i \otimes \phi_j }{\left( P_{\psi_{k}} \otimes I \right) \tilde{U}
    \psi_k \otimes \varphi_l} \\ =& \braket{\left( P_{\psi_{k}} \otimes I
  \right)\psi_i \otimes \phi_j }{ \tilde{U} \psi_k \otimes \varphi_l} = 0,
  \end{align*}
  with $P_{\psi_i} = \ketbra{\psi_i}{\psi_i}$, which agrees with $
  \braket{\psi_i \otimes \phi_j }{\left(I \otimes W^* U V \right)
  \psi_k \otimes \varphi_l}$.

  On the other hand for $i = k$,
  \begin{align*}
    & \braket{\psi_i \otimes \phi_j }{\left(I \otimes W^* U V \right)\psi_i
    \otimes \varphi_l} = \braket{W \phi_j }{ U V \varphi_l} \\ = & \braket{\nu
    \otimes \phi_j }{U \left(\nu \otimes \varphi_l\right)} = \braket{\nu
    \otimes \phi_j }{\left( U^*_i \otimes I \right) U\left(\psi_i \otimes
    \varphi_l\right)}  \\ =& \braket{\left( U_i \otimes I\right) \nu \otimes
    \phi_j }{ U\left(\psi_i \otimes \varphi_l\right)} = \braket{\psi_i \otimes
    \phi_j }{ U\left(\psi_i \otimes \varphi_l\right)},
  \end{align*}
  where $U_i$ is some unitary operator such that $U_i \nu = \psi_i$. Then we
  have $\tilde{U} = I \otimes U_E$ and $U_E$ is an isometry since $\tilde{U}$
  is an isometry. The claim $V_2 = \left( I_A \otimes U_E \right) V_1$ now
  follows by setting $F=I_A$ in Equation \eqref{eq::tilde_u}.
\end{proof}

We now have the sufficient preliminaries to prove Theorem \ref{theorem::theorem}.

\begin{proof}[Proof of Theorem \ref{theorem::theorem}]
  Since $\xi$ is a pure state, there is some unit vector $\nu \in
  \mathcal{H}_B$ such that $\xi = \ketbra{\nu}{\nu}$. Define the linear map
  $V_\nu: \mathcal{H}_A \to \mathcal{H}_A \otimes \mathcal{H}_B$ as $V_\nu \psi
  = \psi \otimes \nu$. For any $\psi, \phi \in \mathcal{H}_A$, $\varphi \in
  \mathcal{H}_B$
$$
  \braket{\varphi \otimes \phi}{V_\nu \psi} = \braket{\varphi
  }{\psi}\braket{\phi}{\nu} = \braket{V_\nu^* \left( \varphi \otimes \phi
  \right)}{\psi}.
$$
  The action of $V^*_\nu$ is then $\varphi \otimes \phi \mapsto
  \braket{\nu}{\phi}\varphi$, and we have the identities $V^*_\nu V_\nu =I$ and
  $V^*_\nu F_A V_\nu = F_A \otimes \xi$. For any linear map $F_A$ on
  $\mathcal{H}_A$. In particular, $V^*_\nu V_\nu = I \otimes \xi$ and  $V^*_\nu
  \rho V_\nu = \rho \otimes \xi$ for any state $\rho$ on $\mathcal{H}_A$.

  Consider a Stinespring dilation $\left\langle \mathcal{H}_{ABE}, \pi_{E}, V
  \right\rangle$ of $\mathcal{N}^*$. Then for all linear maps $F_A$ and $F_B$
  on $\mathcal{H}_A$ and $\mathcal{H}_B$, respectively,
\begin{equation}
  \label{eq::otimes_id}
\mathcal{N}^* \left( F_A \otimes F_B \right) = V^* \left( F_A \otimes F_B \otimes
  I_E \right) V.
\end{equation}
  For any state $\rho$ on $\mathcal{H}_A$ and a linear map
  $F_A$ on $\mathcal{H}_A$
  \begin{align*}
    & \tr \left[\rho \mathcal{N}^*_A \left(F_A  \right)  \right] \\
    = & \tr \left[ \mathcal{N}_A \left( \rho \right) F_A \right] \\
    = & \tr \left[ \tr_B \left[\mathcal{N} \left( \rho \otimes \xi
    \right)\right] F_A\right] \\
    = & \tr \left[ \mathcal{N} \left( \rho \otimes \xi \right) \left( F_A \otimes
    I_B \right) \right] \\
    = & \tr \left[  \left( \rho \otimes \xi \right) \mathcal{N}^* \left(  F_A
    \otimes I_B\right) \right] \\
    = & \tr \left[  \left( V_\nu \rho V^*_\nu \right) V^* \left(  F_A \otimes
    I_{BE}\right) V \right] \\
    = & \tr \left[  \rho \left( V V_\nu\right)^* \left(  F_A \otimes
    I_{BE}\right) \left( V V_\nu \right) \right]. \\
  \end{align*}
  Since the equalities are valid for any $\rho$, this implies
  $$
  \mathcal{N}^*_A \left( F_A \right) = \left( V V_\nu\right)^* \left(  F_A \otimes
    I_{BE}\right) \left( V V_\nu \right).
  $$
  This means, that, given a Stinespring dilation $\left\langle
  \mathcal{H}_{ABE}, \pi_E, V \right\rangle$ of $\mathcal{N}^*$, we can produce
  a Stinespring dilation $\left\langle \mathcal{H}_{ABE}, \pi_{BE}, \tilde
  V_\nu \right\rangle$ of $\mathcal{N}_A^*$ by setting
$$
\tilde V_\nu = V V_\nu.
$$
Assume the channel $\mathcal{N}^*_A$ is unitary $F_A \mapsto U_A^* F_A U_A$.
Then, we can trivially construct a dilation of $\mathcal{N}^*_A$ as
$\left\langle \mathcal{H}_{A} \otimes \mathbb{C}, \pi_{\mathbb{C}}, \tilde
U_A\right\rangle$ by setting
  $$
  \tilde U_A \psi = U_A \psi \otimes 1_\mathbb{C}.
  $$
  Since the vectors $U_A \psi \otimes 1_\mathbb{C}$ clearly span $\mathcal{H}_A
  \otimes \mathbb{C}$, this is a minimal dilation of $\mathcal{N}^*_A$. By
  Lemma \ref{lemma::decompose}, there is then an isometry $U_\nu : \mathbb{C}
  \to \mathcal{H}_{BE}$ such that
$$
\tilde V_\nu = \left( I_A \otimes U_\nu \right) \tilde U_A.
$$
In particular, this implies the following decomposition property for $V$:
$$
V V_\nu = \left( I_A \otimes U_\nu \right) \tilde U_A.
$$

  Denote the dimension of $\mathcal{H}_A$ by $d$ and let $\mathcal{N}^*_B
  \left( F_B \right) = \tr_A \left[ \mathcal{N}^* \left( I_A \otimes F_B
  \right) \right] / d$. Note that
  $$
  \mathcal{N}^*_B \left( I_B \right) =
  \frac{1}{d}\tr_A \left[
    \mathcal{N}^* \left( I_A \otimes I_B \right) \right] =
  \frac{\tr \left[I_A \right]}{d} I_B = I_B,
  $$
  i.e. this map is unital. It is also linear as a composition of linear maps
  and completely positive as a composition of completely positive maps. It is
  then a channel on $\mathcal{H}_B$. We can also calculate, for any linear map
  $F_B$ on $\mathcal{H}_B$,
\begin{align*}
  & \braket{\nu}{\mathcal{N}^*_B \left( F_B \right) \nu} \\
  =& \tr \left[ \xi \mathcal{N}_B^* \left(F_B \right)\right]
  \\
  = & \frac{1}{d} \tr \left[ \xi \tr_A \left[ \mathcal{N}^* \left(I_A
  \otimes F_B \right)\right] \right]
  \\
  = & \frac{1}{d} \tr \left[ \left( I_A \otimes \xi \right) \mathcal{N}^*
  \left(I_A \otimes F_B \right)\right]
  \\
  = & \frac{1}{d} \tr \left[ V_\nu V_\nu^* \mathcal{N}^* \left(I_A
  \otimes F_B \right)\right]
  \\
  = & \frac{1}{d} \tr \left[\left( V V_\nu\right)^*\left(I_A \otimes F_B
  \otimes I_E \right)\left(V V_\nu \right) \right]
  \\
  = & \frac{1}{d} \tr \left[\left(  \left( I_A\otimes  U_\nu\right) \tilde
  U_A\right)^*\left(I_A \otimes F_B \otimes I_E \right)\left(\left( I_A \otimes
  U_\nu\right)\tilde U_A  \right) \right]
  \\
  = & \frac{1}{d} \tr \left[\left(  I_A \otimes  U_\nu \right)^*\left(I_A
  \otimes F_B \otimes I_E \right)\left( I_A\otimes U_\nu  \right) \right]
  \\
  = & \frac{1}{d} \tr \left[ I_A \right] \tr \left[ U_\nu^*\left( F_B
  \otimes I_E \right)U_\nu\right] \\
  =& \braket{1_\mathbb{C}}{U_\nu^*\left( F_B \otimes I_E \right)U_\nu
  1_\mathbb{C}}.
\end{align*}

Then for any unit vector $\psi \in \mathcal{H}_A$ and linear maps $F_A$, $F_B$
on $\mathcal{H}_A$ and $\mathcal{H}_B$, respectively,
\begin{align*}
  & \braket{\psi \otimes \nu}{ \mathcal{N}^* \left( F_A \otimes F_B \right) \psi
  \otimes \nu} \\
  =& \braket{\psi \otimes \nu}{\left( V^* \left( F_A \otimes F_B \otimes I_E
  \right) V \right) \psi \otimes \nu} \\
  = & \braket{V_\nu \psi}{\left( V^* \left( F_A \otimes F_B \otimes I_E
  \right) V \right) V_\nu \psi} \\
  = & \braket{\psi}{\left( \left(V V_\nu \right)^* \left( F_A \otimes F_B \otimes I_E
  \right) \left( V V_\nu \right) \right) \psi} \\
  = & \braket{\psi}{\left( \left( \left( I_A \otimes U_\nu \right) \tilde
  U_A \right)^* \left( F_A \otimes F_B \otimes I_E \right) \left(\left( I_A
  \otimes U_\nu \right) \tilde U_A \right) \right) \psi} \\
  = & \braket{U_A \psi \otimes 1_\mathbb{C}}{\left( \left( I_A \otimes
  U_\nu \right)^* \left( F_A \otimes F_B \otimes I_E \right) \left( I_A
  \otimes U_\nu \right) \right) \left( U_A \psi \otimes
  1_\mathbb{C}\right)} \\
  =&\braket{\psi}{U_A^* F_A U_A \psi}
  \braket{1_\mathbb{C}}{\left(U_\nu^*\left( F_B \otimes I_E \right)
  U_\nu \right) 1_\mathbb{C})} \\
  =&\braket{\psi}{\mathcal{N}_A^* \left( F_A \right)
  \psi}\braket{\nu}{\mathcal{N}_B^* \left( F_B \right) \nu} \\
  = & \braket{\psi \otimes \nu}{ \left( \mathcal{N}_A^* \left( F_A \right)
  \otimes \mathcal{N}_B^* \left( F_B \right) \right) \psi \otimes \nu}.
\end{align*}
And 
$$
\tr \left[\left( P_\psi \otimes \xi \right) \mathcal{N}^* \left( F_A
\otimes F_B \right) \right] = \tr \left[\left( P_\psi \otimes \xi \right) \left(
\mathcal{N}_A^* \left( F_A \right) \otimes \mathcal{N}^*_B \left( F_B \right)
\right) \right],
$$
or, in the dual picture,
$$
\tr \left[\mathcal{N}\left( P_\psi \otimes \xi \right)  \left( F_A \otimes
F_B \right) \right] = \tr \left[\left(\mathcal{N}_A \left( P_\psi \right) \otimes
\mathcal{N}_B \left( \xi \right)\right) \left( F_A \otimes F_B \right) \right].
$$
Since this is holds for any choice of $F_A$ and $F_B$, we have for any pure
state $P_\psi$
$$
\mathcal{N}\left( P_\psi \otimes \xi \right) = \mathcal{N}_A \left( P_\psi
\right) \otimes \mathcal{N}_B \left( \xi \right) .
$$
Then, expressing a general state $\rho$ as a convex combination of pure states
$P_i$,
$$
\mathcal{N}\left( \rho \otimes \xi \right) = \sum_i \lambda_i \mathcal{N}\left(
P_i \otimes \xi \right) = \sum_i \lambda_i \mathcal{N}_A\left( P_i \right)
\otimes \mathcal{N}_B \left( \xi \right) = \mathcal{N}_A\left( \rho \right)
\otimes \mathcal{N}_B \left( \xi \right).
$$
\end{proof}

\begin{proof}[Proof of Theorem \ref{theorem::no_conditional}]
  Let $\nu_1, \nu_2 \in \mathcal{H}_A$ be unit vectors so that $\xi_1 =
  \ketbra{\nu_1}{\nu_1}$ and $\xi_2 = \ketbra{\nu_2}{\nu_2}$. Define the
  linear maps $V_{\nu_1}, V_{\nu_2}: \mathcal{H}_A \to \mathcal{H}_A \otimes
  \mathcal{H}_B$ as $V_{\nu_i} \psi = \psi \otimes {\nu_i}$ for $i = 1, 2$.

  Assume both of the channels $\mathcal{N}_{A,i}$ are unitary, and let
  $U_{A, i}$ be the corresponding unitary operators. Define
  $$
  \tilde U_{A, i} \psi = U_{A, i} \psi \otimes 1_\mathbb{C}.
  $$
  Consider a Stinespring dilation $\left\langle \mathcal{H}_{ABE}, \pi_{E}, V
  \right\rangle$ of $\mathcal{N}^*$. It was explained in the proof of Theorem
  \ref{theorem::theorem} how to produce a Stinespring dilation $\left\langle
  \mathcal{H}_{ABE}, \pi_{BE}, \tilde V_{\nu_i} \right\rangle$ of $\mathcal{N}_A^*$
  by setting
$$
  \tilde V_\nu = V V_{\nu_i} = \left( 1 \otimes U_{\nu_i}\right) \tilde U_{A, i}.
$$
We then get the formula
$$
  \braket{\nu_1}{\nu_2} I_A = V_{\nu_1}^*V^*VV_{\nu_2} = \tilde{U}^*_{A,
  1}\left(I_A \otimes U_{\nu_1}^* U_{\nu_2} \right)\tilde{U}_{A, 2}.
$$
Take $\nu_1, \nu_2$ non-orthogonal, and this rearranges to
$$
  \tilde U_{A, 1} = \frac{U_{\nu_1}^* U_{\nu_2}}{\braket{\nu_1}{\nu_2}}
  \tilde{U}_{A, 2} = \lambda \tilde{U}_{A, 2}.
$$
with the complex scalar $\lambda = U_{\nu_1}^* U_{\nu_2} /
  \braket{\nu_1}{\nu_2}$. Since $\tilde U_{A, 1}^* \tilde U_{A, 1} = I$,
  $\left| \lambda \right| = 1$, i.e. the operators $U_{A, 1}$ and $U_{A, 2}$
  are equal up to a phase factor, and the channels $\rho \mapsto U_{A,
  1}\rho U_{A, 1}^*$ and $\rho \mapsto U_{A, 2}\rho U_{A, 2}^*$
  are equal.
\end{proof}
  
\begin{proof}[Proof of Theorem \ref{theorem::theorem3}]
  Consider non-orthogonal $\xi_1$, $\xi_2$. Since the channels
  $$
  \mathcal{N}_{A, \xi_i} \left( \rho \right) = \tr_B \left[ \mathcal{N} \left(
  \rho \otimes \xi_i \right) \right],
  $$
  are unitary, by Theorem \ref{theorem::theorem}
$$
  \mathcal{N} \left( \rho \otimes \xi_i \right) = \mathcal{N}_{A, i} \left(
  \rho \right) \otimes \mathcal{N}_B \left( \xi_i \right),
$$
  and, since $\xi_1$ and $\xi_2$ are not orthogonal, by Theorem
  \ref{theorem::no_conditional}, $\mathcal{N}_{A, 1} = \mathcal{N}_{A, 2}$. For
  orthogonal $\xi_1, \xi_2$ fix a $\xi_3$ that is not orthogonal to either to
  them, to get $\mathcal{N}_{A, 1} =\mathcal{N}_{A, 3} = \mathcal{N}_{A, 2}$.
  There is then a fixed unitary channel $\mathcal{N}_A$ such that
  $$
  \mathcal{N} \left( \rho \otimes \xi \right) =  \mathcal{N}_A \left( \rho
  \right) \otimes \mathcal{N}_B \left( \xi \right),
  $$
  for all pure states $\xi$ on $\mathcal{H}_B$. Express a general state $\sigma$
  as a convex combination of pure states to get
  $$
  \mathcal{N} \left( \rho \otimes \sigma \right) = \sum_i \lambda_i \mathcal{N}
  \left( \rho \otimes \xi_i \right) = \sum_i \lambda_i\mathcal{N}_A \left( \rho
  \right) \otimes \mathcal{N}_B \left( \xi_i \right)= \mathcal{N}_A \left( \rho
  \right) \otimes \mathcal{N}_B \left( \sigma \right).
  $$
  For any bipartite state $\tau$, there are \cite{Beckman2001}, not
  necessary positive $\mu_i$, such that $\tau = \sum_i \mu_i \rho_i \otimes
  \sigma_i$. Then
  $$
  \mathcal{N} \left( \tau \right) = \sum_i \mu_i \mathcal{N} \left( \rho_i
  \otimes \sigma_i \right) = \sum_i \mu_i\mathcal{N}_A \left( \rho_i \right)
  \otimes \mathcal{N}_B \left( \sigma_i \right)= \left(\mathcal{N}_A \otimes
  \mathcal{N}_B \right) \left( \tau \right).
  $$

\end{proof}

\end{document}